\documentclass[journal]{IEEEtran}

\usepackage[english]{babel}
\usepackage{ifpdf}

\usepackage{cite} 
\usepackage{url}
\usepackage{hyperref}

\ifCLASSINFOpdf
	\usepackage[pdftex]{graphicx}
	\graphicspath{{./figures/}}
\else
	\usepackage[dvips]{graphicx}
	\graphicspath{./figures/}
\fi
\usepackage{color}
\usepackage{pgf, tikz, pgfplots}
\usetikzlibrary{shapes, arrows, automata}
\usetikzlibrary{calc,hobby,decorations}

\usepackage[cmex10]{amsmath}
\usepackage{amsfonts, amssymb, amsthm}
\usepackage{mathrsfs}

\usepackage{algorithm} 
\usepackage{algorithmic} 
\usepackage{amsmath} 
\usepackage{xcolor}
\usepackage{enumerate}
\usepackage{multirow}
\usepackage{rotating}
\usepackage{subcaption}
\usepackage{needspace}

\input{mySymbol.sty}

\definecolor{penndarkestblue}{cmyk}{1,0.74,0,0.77}
\definecolor{penndarkerblue}{cmyk}{1,0.74,0,0.70}
\definecolor{pennblue}{cmyk}{0.99,0.66,0,0.57} 
\definecolor{pennlighterblue}{cmyk}{0.98,0.44,0,0.35}
\definecolor{pennlightestblue}{cmyk}{0.38,0.17,0,0.17} 

\definecolor{penndarkestred}{cmyk}{0,1,0.89,0.66}
\definecolor{penndarkerred}{cmyk}{0,1,0.88,0.55}
\definecolor{pennred}{cmyk}{0,1,0.83,0.42} 
\definecolor{pennlighterred}{cmyk}{0,1,0.6,0.24}
\definecolor{pennlightestred}{cmyk}{0,0.43,0.26,0.12} 

\definecolor{penndarkestgreen}{cmyk}{1,0,1,0.68}
\definecolor{penndarkergreen}{cmyk}{1,0,1,0.57}
\definecolor{penngreen}{cmyk}{1,0,1,0.44} 
\definecolor{pennlightergreen}{cmyk}{1,0,1,0.25}
\definecolor{pennlightestgreen}{cmyk}{0.43,0,0.43,0.13}

\definecolor{penndarkestorange}{cmyk}{0,0.65,1,0.49}
\definecolor{penndarkerorange}{cmyk}{0,0.65,1,0.33}
\definecolor{pennorange}{cmyk}{0,0.54,1,0.24} 
\definecolor{pennlighterorange}{cmyk}{0,0.32,1,0.13}
\definecolor{pennlightestorange}{cmyk}{0,0.15,0.46,0.06}
	
\definecolor{penndarkestpurple}{cmyk}{0,1,0.11,0.86}
\definecolor{penndarkerpurple}{cmyk}{0,1,0.13,0.82}
\definecolor{pennpurple}{cmyk}{0,1,0.11,0.71} 
\definecolor{pennlighterpurple}{cmyk}{0,1,0.05,0.46}
\definecolor{pennlightestpurple}{cmyk}{0,0.35,0.02,0.23}
	
\definecolor{pennyellow}{cmyk}{0,0.20,1,0.05} 
\definecolor{pennlightgray1}{cmyk}{0,0,0,0.05}
\definecolor{pennlightgray3}{cmyk}{0.01,0.01,0,0.18}
\definecolor{pennmediumgray1}{cmyk}{0.04,0.03,0,0.31}
\definecolor{pennmediumgray4}{cmyk}{0.08,0.06,0,0.54}
\definecolor{penndarkgray2}{cmyk}{0.09,0.07,0,0.71}
\definecolor{penndarkgray4}{cmyk}{0.1,0.1,0,0.92}

\def\SO3{\mathrm{SO(3)}}

\newtheorem{assumption}{\hspace{0pt}\bf Assumption}

\newtheorem{proposition}{\hspace{0pt}\bf Proposition}

\newtheorem{theorem}{\hspace{0pt}\bf Theorem}

\begin{document}

\title{Learning Stable Graph Neural Networks \\ via Spectral Regularization}

\author{\IEEEauthorblockN{Zhan Gao$^{\dagger}$ and Elvin Isufi$^{\ddagger}$\\
\thanks{$^{\dagger}$Department of Computer Science and Technology, University of Cambridge, Cambridge, UK. $^{\ddagger}$Department of Intelligent Systems, Delft University of Technology, Delft, The Netherlands. The work of E. Isufi is supported by the TU Delft AI Labs programme. 
	Email: zg292@cam.ac.uk, e.isufi-1@tudelft.nl }}
}


\maketitle

\begin{abstract}
	
Stability of graph neural networks (GNNs) characterizes how GNNs react to graph perturbations and provides guarantees for architecture performance in noisy scenarios. This paper develops a self-regularized graph neural network (SR-GNN) solution that improves the architecture stability by regularizing the filter frequency responses in the graph spectral domain. The SR-GNN considers not only the graph signal as input but also the eigenvectors of the underlying graph, where the signal is processed to generate task-relevant features and the eigenvectors to characterize the frequency responses at each layer. We train the SR-GNN by minimizing the cost function and regularizing the maximal frequency response close to one. The former improves the architecture performance, while the latter tightens the perturbation stability and alleviates the information loss through multi-layer propagation. We further show the SR-GNN preserves the permutation equivariance, which allows to explore the internal symmetries of graph signals and to exhibit transference on similar graph structures. Numerical results with source localization and movie recommendation corroborate our findings and show the SR-GNN yields a comparable performance with the vanilla GNN on the unperturbed graph but improves substantially the stability.

\end{abstract}


\IEEEpeerreviewmaketitle

\section{Introduction} \label{sec:intro}

Data generated by networks, such as sensor, social and multi-agent systems, exhibits an irregular structure inherent in its underlying topology \cite{marsden1990network, aggarwal2011introduction}. 
Graph neural networks (GNNs) leverage this topological information to model task-relevant representations from networked data \cite{Wu2019, gama2020graphs, gao2021variance, isufi2021edgenets} with applications in wireless communications \cite{gao2020resource, gao2022decentralized}, recommender systems \cite{gao2022graph, isufi2021accuracy} and robot swarm control \cite{tolstaya2020learning, gao2022wide}. The GNN leverages the underlying graph as an inductive bias to learn representations and also as a platform to process information distributively. However, such a graph is often perturbed by external factors such as channel fading, adversarial attacks, or estimation errors \cite{xu2004wireless, Gungor2010}. In these cases, GNNs will be trained on the underlying graph but implemented on the perturbed one, resulting in a degraded performance.

Characterizing the GNN stability under graph perturbations has been investigated in \cite{levie2019transferability, gama2020stability, gao2021stability, kenlay2021stability}. In particular, \cite{levie2019transferability} showed that GNNs model similar representations on graphs that describe the same phenomenon. The work in \cite{gama2020stability} characterized the GNN stability to absolute and relative perturbations, which indicates that the GNN can be both stable and discriminative. The authors in \cite{gao2021stability} analyzed the stability under stochastic perturbations and identified the effects of graph stochasticity, filter property and architecture hyperparameters. Moreover, the work in \cite{kenlay2021stability} focused on structural perturbations and established a stability bound that has structural interpretations. These results indicate that GNNs can maintain performance under small perturbations, but suffer from inevitable degradation under large ones. 

To alleviate the performance degradation under support perturbations, the work in \cite{gao2021stochastic} developed stochastic graph neural networks that account for graph stochasticity during training and improves robustness to random link losses, while the authors in \cite{gao2021training} trained robust GNNs by using topological adaptive edge dropping as a data augmentation technique. However, these approaches focus only on perturbations related to random link losses. The authors in \cite{cervino2021training} proposed improving stability by imposing an extra constraint on the loss function that restricts the GNN output deviations induced by graph perturbations, but it requires the perturbation information to characterize output deviations during training that may not always be available. Authors in \cite{arghal2021robust} followed a similar training procedure but constrained the filter Lipschitz constant with the scenario approach. While improving the stability, deploying such approaches requires a good intuition on the selection of the constraint bound. A larger bound has less effects on performance but results in a looser stability, while a lower bound yields a tighter stability but degrades the performance. The latter requires hand-tuning at the outset, which could be sensitive to specific scenarios and time-consuming. 

To overcome these issues, we observe that lower values of the filter frequency response yield more stable GNNs but result in an information loss for data propagation through layers \cite{gama2020stability}. This indicates that frequency responses being close to one would be a favorable trade-off because it does not explode the stability bound and maintains the signal information. We leverage this observation to improve stability without requiring any additional prerequisite knowledge, e.g., the perturbation information and the constraint bound selection. Most similar to our approach is \cite{gama2020stability1}, which trains the GNN by regularizing the frequency response directly without considering the information propagation. The latter may lead to an information loss and degrade the performance.

This paper develops the self-regularized graph neural network (SR-GNN) architecture that balances stability with performance by regularizing the filter frequency responses close to one (Section \ref{sec:SRGNN}). The SR-GNN takes both the data and the eigenvectors of the underlying graph as inputs. The former is processed to extract task-relevant features, while the latter are leveraged to characterize the frequency responses of each layer. We train the SR-GNN by regularizing the task-relevant cost with a term that pushes the maximal frequency response close to one, pursuing a trade-off between stability and information propagation (Section \ref{subsec:training}). The SR-GNN is permutation equivariant that allows robust transference (Section \ref{subsec:permutationEquivariance}). We corroborate the SR-GNN in experiments on the source localization and the movie recommendation in Section \ref{sec:experiment}, and conclude the paper in Section \ref{sec_conclusions}.


\section{Preliminaries} \label{sec:StoGNN}

Let $\ccalG = \{\ccalV,\ccalE\}$ be an undirected graph with node set $\ccalV = \{1,...,n\}$ and edge set $\ccalE = \{(i,j)\}\subset \ccalV \times \ccalV$. The graph is associated to a shift operator matrix $\bbS \in \mathbb{R}^{n \times n}$ that captures its sparsity with $[\bbS]_{ij} \ne 0$ iff $(i,j) \in \ccalE$ or $i=j$, such as the adjacency matrix $\bbA$ and the graph Laplacian $\bbL$. Define the graph signal as a vector $\bbx = [x_1,...,x_n]^\top \in \mathbb{R}^n$, where the $i$th entry $x_i$ is the signal value assigned to node $i$. Learning tasks with graphs and signals consider devising solutions capable of capturing the interplay between these two.

\smallskip
\noindent \textbf{Graph neural network (GNN).} A GNN is a layered architecture, where each layer comprises a graph filter bank and a pointwise nonlinearity. At layer $\ell$, the inputs are $F$ graph signals $\{\bbx_\ell^{g}\}_{g=1}^F$ generated at layer $(\ell\!-\!1)$. These input signals are processed by graph filters $\bbH_\ell^{fg}(\bbS)$, which are polynomials of the shift operator that aggregate multi-hop neighborhood information, i.e.,
\begin{align}\label{eq:graphFilterOutput}
	\bbu_\ell^{fg} \!=\! \bbH_\ell^{fg}(\bbS) \bbx_{\ell-1}^g \!:=\! \sum_{k=0}^K h_{k\ell}^{fg} \bbS^k \bbx
\end{align}
for all $g,f =1,\ldots,F$, where $\bbh_{\ell}^{fg}=\{h_{k\ell}^{fg}\}_{k=0}^K$ are the filter coefficients. Expression \eqref{eq:graphFilterOutput} indicates that each of $F$ input signals $\bbx_\ell^g$ is processed by a bank of $F$ graph filters to generate $F^2$ output signals $\bbu_\ell^{fg}$. The latter are summed over the input index $g$ and passed through a pointwise nonlinearity to generate $F$ outputs of layer $\ell$ as
\begin{align}\label{eq:GNN}
	\bbx_\ell^f = \sigma\Big( \sum_{g=1}^F \bbu_\ell^{fg} \Big)
\end{align}
for all $f=1\ldots,F$ and $\ell=1,\ldots,L$. We consider a single input $\bbx_0^1 = \bbx$ and a single output $\bbx_L^1$ and write the GNN as a nonlinear map $\bbPhi(\bbx;\bbS,\ccalH):\mathbb{R}^n \to \mathbb{R}^n$ where $\ccalH = \{\bbh_{\ell}^{fg}\}_{\ell f g}$ are the architecture parameters \cite{gama2020graphs}.

%
\begin{figure}
	\centering
	\includegraphics[width=0.9\columnwidth]
	{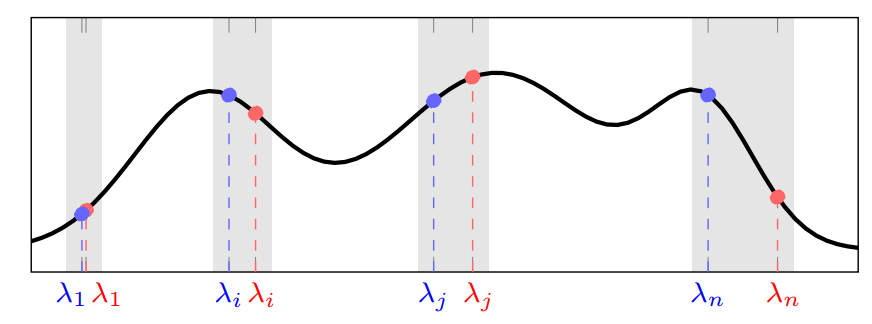}
	\caption{The filter frequency response $h(\lambda)$ (black line). The function $h(\lambda)$ is analytic w.r.t. the generic variable $\lambda$ [cf. \eqref{eq:filterFrequency}]. Specific graphs instantiate specific eigenvalues $\{\lambda_i\}_{i=1}^n$/$\{\lambda_j\}_{j=1}^n$ on $\lambda$ (red/blue points).} \label{fig:frequencyResponse}
\end{figure}
%

\smallskip
\noindent\textbf{Frequency response.} One particularity of GNNs is that their filters admit a spectral equivalence by leveraging concepts from graph signal processing \cite{ortega2018}. Specifically, let $\bbS = \bbV \bbLambda \bbV^{\top}$ be the eigendecomposition with eigenvectors $\bbV = [\bbv_{1}, \cdots, \bbv_{n}]$ and eigenvalues $\bbLambda = \text{diag} (\lambda_{1},...,\lambda_{n})$. By means of the graph Fourier transform (GFT), we can expand the signal as $\bbx = \sum_{i=1}^n \hat{x}_i \bbv_{i}$ with $\hat{\bbx}\!=\![\hat{x}_1, \cdots, \hat{x}_n]^\top$ the Fourier coefficients. Substituting the latter into the filter output \eqref{eq:graphFilterOutput}, we have the input-output relation
\begin{equation} \label{eq:graphFilterSpectralOutput}
	\begin{split}
		\bbu \!=\! \sum_{k=0}^K h_k \bbS^k \sum_{i=1}^n \hat{x}_i \bbv_{i} \!=\! \sum_{i=1}^n \sum_{k=0}^K \hat{x}_i h_k \lambda_{i}^k \bbv_{i} .
	\end{split}
\end{equation}
With the similar expansion of the output $\bbu = \sum_{i=1}^n \hat{u}_i \bbv_i$, the spectral input-output filtering relation is $\hat{\bbu} = \bbH(\bbLambda)\hat{\bbx}$. Here, $\bbH(\bbLambda) = \text{diag}(\sum_{k=0}^K h_k \lambda_1^k,\ldots, \sum_{k=0}^K h_k \lambda_n)$ is a diagonal matrix containing the filter frequency responses. This shows the filter coefficients $\{h_k\}_{k=0}^K$ define the frequency response and the underlying graph $\bbS$ only instantiates the specific eigenvalues $\{\lambda_i\}_{i=1}^n$. We can then characterize the graph filter via the analytic function
\begin{equation} \label{eq:filterFrequency}
	\begin{split}
		h(\lambda) = \sum_{k=0}^K h_k \lambda^k
	\end{split}
\end{equation}
where $\lambda$ is the function variable -- see Fig. \ref{fig:frequencyResponse}.

\section{Analysis and Motivation}\label{sec:Analysis}

%
\begin{figure*}
	\centering
	\includegraphics[width=1.8\columnwidth]
	{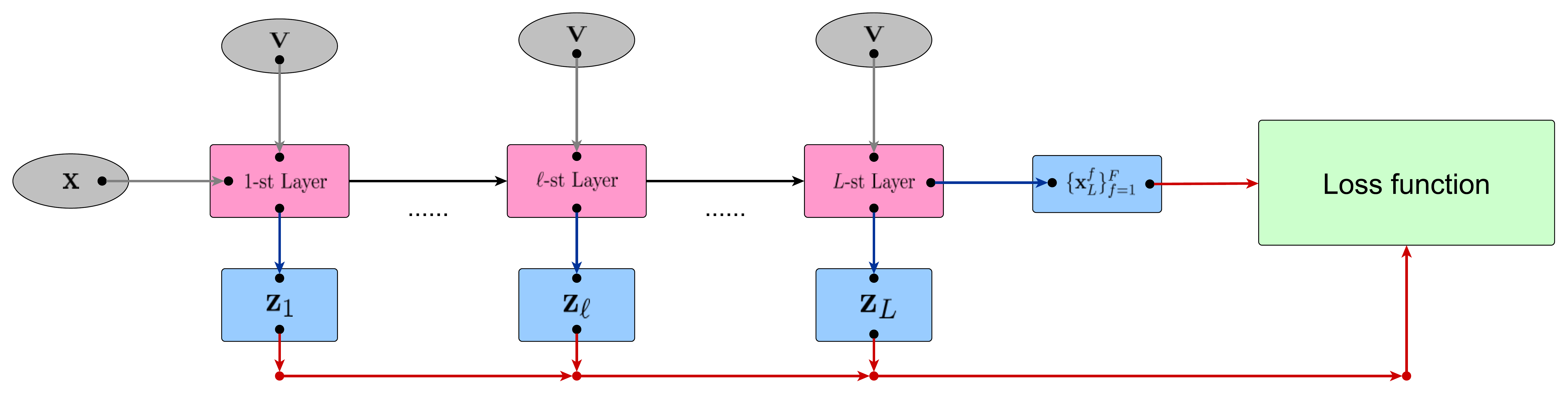}
	\caption{The self-regularized graph neural network (SR-GNN). The inputs are the graph signal $\bbx$ and the eigenvectors $\bbV$ of the underlying graph $\bbS = \bbV\bbLambda\bbV^\top$. The outputs are the task-relevant features $\{\bbx^f_L\}_{f=1}^F$ and the spectral outputs $\{\bbz_\ell\}_{\ell=1}^L$ of $L$ layers. The former is used to compute the task-relevant cost and the latter is used for the spectral regularization, both of which complete the loss function.}%
	\label{fig:SRGNNArchitecture}
\end{figure*}
%

GNNs model representations from graph signals by exploiting the graph topology and thus, may lead to a degraded performance under topological perturbations. The perturbation effects on the GNN performance have been analyzed in \cite{gama2020stability}, which relies on the following assumption.
\begin{assumption}[Bounded Lipschitz filter] \label{as:boundedLispchitzFilter}
	A graph filter with the frequency response $h(\lambda)$ as in \eqref{eq:filterFrequency} is bounded Lipschitz, i.e., there exist constants $C_U$ and $C_L$ such that
	\begin{align}
		|h(\lambda)| \le C_U~\text{and}~ |h'(\lambda)| \le C_L,~\forall~\lambda \in \mathbb{R}.
	\end{align}
\end{assumption}
%
\noindent Assumption \ref{as:boundedLispchitzFilter} indicates that the filter frequency response is bounded and does not change faster than linear in the graph spectral domain. The following theorem then formalizes the stability of GNNs.
\begin{theorem}[Stability of GNNs \cite{gama2020stability}] \label{theorem:GNNstability}
	Consider the GNN $\bbPhi(\bbx; \bbS, \ccalH)$ of $L$ layers and $F$ features with the underlying graph $\bbS$ [cf. \eqref{eq:GNN}]. Let the perturbed graph $\widetilde{\bbS}$ be such that $\|\widetilde{\bbS} - \bbS \| \le \eps$ with $\eps$ the perturbation size. Let also the graph filters satisfy Assumption \ref{as:boundedLispchitzFilter} and the nonlinearity be normalized Lipschitz, i.e., $|\sigma(a) - \sigma(b)| \le |a - b|$ for all $a,b\in \mathbb{R}$. Then, for any graph signal $\bbx$, it holds that
	\begin{align}\label{eq:GNNstability}
		&\| \bbPhi(\bbx;\bbS,\ccalH) - \bbPhi(\bbx;\widetilde{\bbS},\ccalH)\| \le C \| \bbx \| \eps \!+\! \mathcal{O}(\eps^2)
	\end{align}
	where $C = (1+8\sqrt{n}) L C_L (C_U F)^L$ is the stability constant.
\end{theorem}

\noindent Theorem \eqref{theorem:GNNstability} states that the output difference of the GNN is upper bounded proportionally by the size of the perturbation $\eps$ w.r.t. the stability constant $C$. The bound reduces to zero when $\eps \to 0$, i.e., the GNN maintains performance under mild perturbations.

\smallskip
\noindent \textbf{Problem motivation.} GNNs degrade substantially when the perturbation size $\eps$ is large. The latter motivates to improve stability by acting on the architecture to decrease the stability constant $C$ [cf. \eqref{eq:GNNstability}]. This could be achieved in three ways:

\begin{enumerate}[(i)]
	
	\item \emph{Reducing graph size.} The term $(1+8\sqrt{n})$ indicates that a smaller graph leads to a more stable GNN. However, the underlying graph is determined by problem settings and cannot be changed during training.
	
	\item \emph{Reducing architecture width and depth.} The term $LF^L$ is the consequence of the perturbation passing through filter banks $F$ and layers $L$. The GNN becomes more stable with less features or less layers, but the latter degrades its representational capacity. 
	
	\item \emph{Tuning parameters.} The term $C_L C_U^L$ indicates the impact of the frequency response $h(\lambda)$, which are determined by the architecture parameters $\ccalH$. The stability bound decreases with the bound constant $C_U$ and the Lipschitz constant $C_L$. 
	
\end{enumerate}

\noindent The above observations motivate improving the GNN stability by (iii) tuning the architecture parameters $\ccalH$ to reduce the Lipschitz constant $C_L$ or the bound constant $C_U$. Note also that the stability constant $C$ is proportional to $C_L$ and to $C_U^L$. The latter increases faster than the former especially for multi-layer GNNs, e.g., $L \ge 2$. Therefore, we set focus on developing an architecture that restricts the bound constant $C_U$ to improve stability. 

A smaller $C_U$ yields a more stable GNN but it also leads to an information loss passing through the filter; hence, degrading performance. Since each layer output $\bbx_\ell^f$ is aggregated over $F$ intermediate features $\{\bbu_{\ell}^{fg}\}_{g=1}^F$ [cf. \eqref{eq:GNN}], our expectation is to tune the bound constant $C_U$ such that $C_U F$ is close to one. The latter would be a good trade-off between perturbation stability and information propagation, i.e., the term $(C_U F)^L$ in the stability bound \eqref{eq:GNNstability} does not explode with $L$ and the input passes through the layers with a contained information loss. To do so, we develop the self-regularized graph neural network (SR-GNN) that regularizes the task-relevant cost with a term pushing $C_U F$ close to one at each layer. The SR-GNN is self-contained, where no prerequisite knowledge is required, e.g., the perturbation information, and no outside parameter is introduced, e.g., the constraint bound.

\section{Self-Regularized Graph Neural Networks}\label{sec:SRGNN}

We propose an architecture that not only extracts task-relevant features from graph signals but also generates filter frequency responses at each layer. The frequency responses are used to characterize the bound constant $C_U$ and regularize the latter to improve stability. This architecture pursues a trade-off between performance and stability via a spectral regularization.
 
For a specific task, the underlying graph $\bbS$ instantiates the specific eigenvalues $\{\lambda_i\}_{i=1}^n$ on the filter frequency response $h(\lambda)$ [cf. \eqref{eq:filterFrequency}]. Hence, Theorem \ref{theorem:GNNstability} implies that the GNN is stable w.r.t. this graph $\bbS$ as long as the frequency responses at these eigenvalues $\{h(\lambda_i)\}_{i=1}^n$ are bounded Lipschitz, i.e., $|h(\lambda_i)| \le C_U$ for $i=1,...,n$, instead of at the entire graph frequency domain. This indicates that we only need to focus on $\{h(\lambda_i)\}_{i=1}^n$ during training, which can be characterized by leveraging the eigenvectors $\bbV =\{\bbv_i\}_{i=1}^n$ together with the spectral input-output filtering relation \eqref{eq:graphFilterSpectralOutput}.

\smallskip
\noindent\textbf{Self-regularized graph neural network (SR-GNN).} The SR-GNN follows the same layered architecture as the vanilla GNN, which takes $F$ graph signals $\{\bbx_{\ell-1}^{g}\}_{g=1}^F$ and the eigenvectors $\bbV$ of $\bbS$ as inputs at layer $\ell$. Both are processed by the same bank of filters with shared parameters $\{\bbH_\ell^{fg}(\bbS)\}_{fg}$ and aggregated over the input index $g$ as
\begin{align}\label{eq:SRGNNOutput}
	&\bbu_\ell^{f} = \sum_{g=1}^F \bbH_\ell^{fg}(\bbS) \bbx_{\ell-1}^g,\\
	&\bbz_{\ell i}^f = \sum_{g=1}^F \sum_{k=0}^K h_{\ell k}^{fg} \bbS^k \bbv_i = \sum_{g=1}^F h_\ell^{fg}(\lambda_i) \bbv_i
\end{align}
for all $f=1,\ldots,F$. The filtered signals $\{\bbu_\ell^{f}\}_{f=1}^F$ are passed through the pointwise nonlinearity $\sigma(\cdot)$ as
\begin{align}\label{eq:SRGNNOutput2}
	\bbx_{\ell+1}^{f} = \sigma\big(\bbu_\ell^{f}\big),~\forall~ f=1,\ldots,F
\end{align}
which are the same higher-level features as the vanilla GNN and are propagated to the next layer $(\ell+1)$. Instead, the outputs $\{\bbz_{\ell,i}\}_{f=1}^F$ are multiplied by the respective eigenvector $\bbv_i$ to compute $F$ frequency responses of layer $\ell$. These are passed through the maxout nonlinearity as 
\begin{align}\label{eq:layerOutput}
	z_{\ell i} \!=\! \max_{f=1,...,F} \!\left(\!(\bbz_{\ell i}^f)^\top\! \bbv_i\! \right) \!=\! \!\max_{f=1,...,F} \!\left(\sum_{g=1}^F\! h_\ell^{fg}(\lambda_i)\!\right)
\end{align}
for all $i=1,...,n$, which is the maximal sum of the $F$ frequency responses on the eigenvalue $\lambda_i$ and thus, corresponds to $F C_U$ in the stability bound \eqref{eq:GNNstability}. The latter are concatenated to a feature vector $\bbz_\ell = [z_{\ell 1},\ldots,z_{\ell n}]^\top$, referred to as the \emph{spectral output} of layer $\ell$. The SR-GNN is a nonlinear map $\bbPhi(\bbx,\bbV;\bbS,\ccalH)$ with inputs the graph signal $\bbx$ and the eigenvectors $\bbV$ and outputs the task-relevant features $\{\bbx_L^f\}_{f=1}^F$ and the spectral outputs of $L$ layers $\{\bbz_{\ell}\}_{\ell\!=\!1}^L$, i.e., $\bbPhi(\bbx,\!\bbV;\bbS,\!\ccalH) \!=\! \big\{\bbPhi(\bbx;\bbS,\!\ccalH), \{\bbz_{\ell}\}_{\ell=1}^L\big\}$ -- see Fig. \ref{fig:SRGNNArchitecture}.

\begin{figure*}[t]
	\centering
	\begin{subfigure}{.325\textwidth}
		\includegraphics[width=\textwidth]{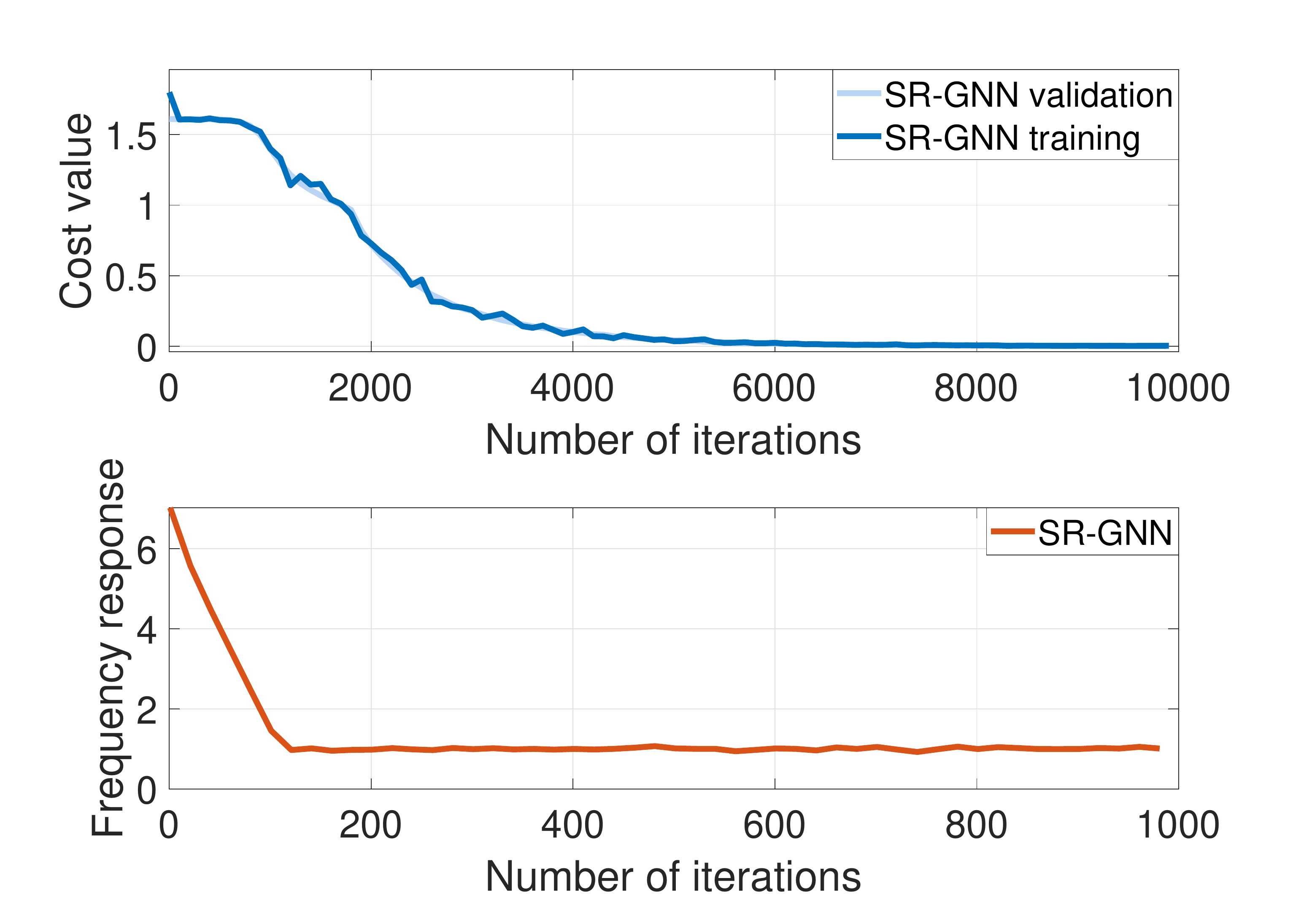}
		\caption{}
		\label{fig:convergence}
	\end{subfigure}
	\begin{subfigure}{.325\textwidth}
		\includegraphics[width=\textwidth]{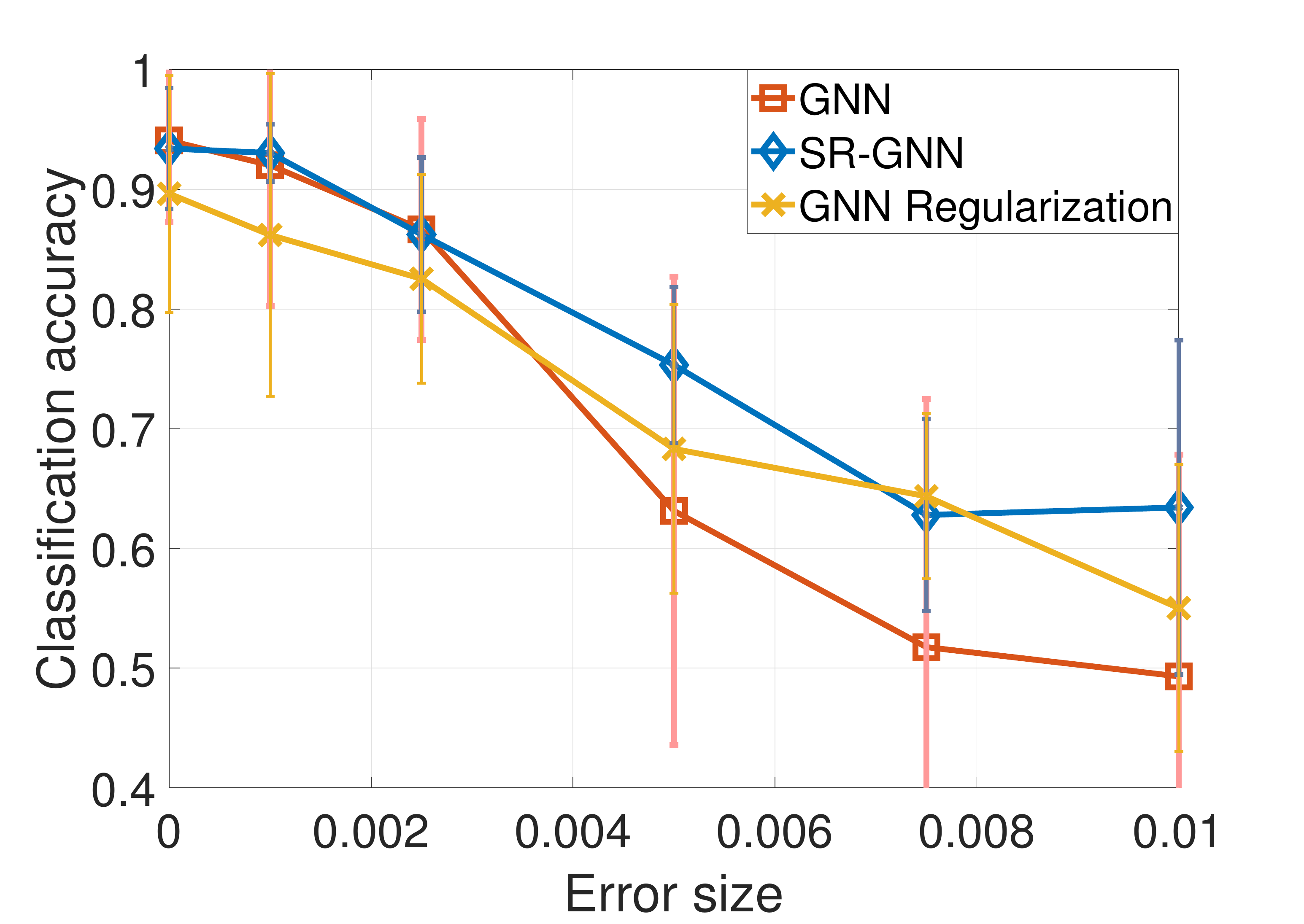}
		\caption{}
		\label{fig:source}
	\end{subfigure}
	\begin{subfigure}{.325\textwidth}
		\includegraphics[width=\textwidth]{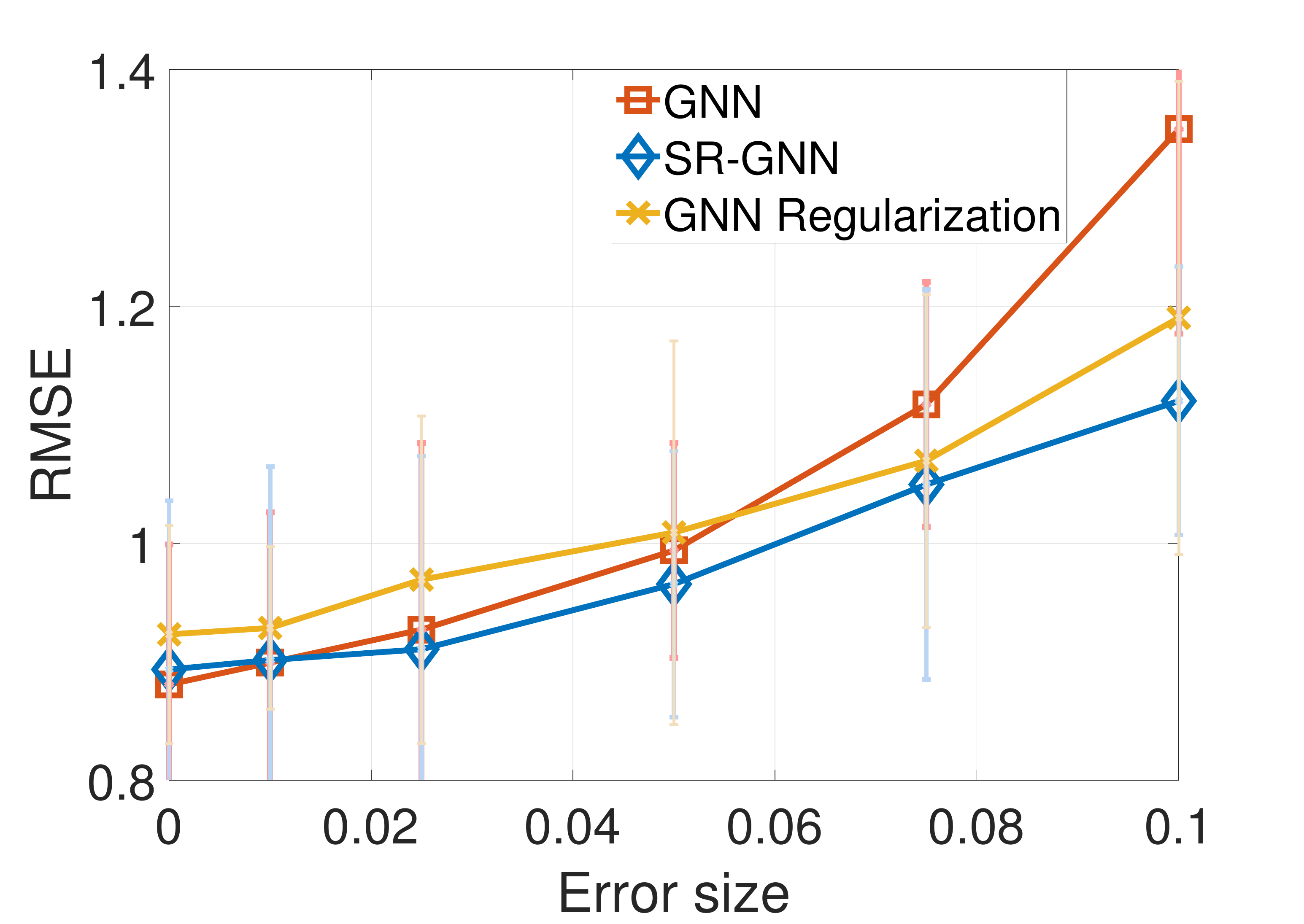}
		\caption{}
		\label{fig:movie}
	\end{subfigure}
	\caption{The performance of the SR-GNN. (a) The convergence of the SR-GNN. (b) The performance comparison in source localization. (c) The performance comparison in movie recommendation.}
	\label{fig:experiment}
\end{figure*}

\subsection{Training}\label{subsec:training}

Let $C(\cdot,\cdot)$ be the cost function and $\ccalR= \{(\bbx_r, \bby_r)\}$ the training set. The goal is to minimize the cost regularized by a term that pushes the maximal value of the spectral output $\bbz_\ell$ close to one at each layer. The cost is the task fitting term, while the regularizer balances the stability with the information propagation. Mathematically, this implies solving the optimization problem 
\begin{align}\label{eq:trainObjective}
	\min_{\ccalH} \ccalC (\ccalR, \gamma; \ccalH) = \min_{\ccalH} \frac{1}{|\ccalR|} \!\!&\sum_{(\bbx_r,\bby_r)\in \ccalR}\! \!\!\!C(\bbPhi(\bbx_r,\bbS;\ccalH), \bby_r) \\
	&+ \frac{\gamma}{L}\sum_{\ell=1}^L |1 - \max_{i=1,\ldots,n}(z_{\ell i})| \nonumber 
\end{align}
where $|\ccalR|$ is the number of data samples, $z_{\ell i}$ is the $i$th entry of the spectral output $\bbz_\ell$ at layer $\ell$ [cf. \eqref{eq:layerOutput}], and $\gamma > 0$ is the regularization parameter that weighs the deviation of the maximal value of the spectral output from one averaged over $L$ layers. The regularizer purses a trade-off between the stability to topological perturbations and the information propagation through multiple layers.

We solve problem \eqref{eq:trainObjective} with stochastic gradient descent. At each iteration $t$, the architecture parameters $\ccalH_t$ are updated as
\begin{align}\label{eq:SGD}
	\ccalH_{t+1} = \ccalH_t - \alpha_t \nabla_\ccalH \ccalC(\ccalR_t, \gamma; \ccalH_t)
\end{align}
where $\ccalR_t \in \ccalR$ is a subset of $\ccalR$ randomly selected at iteration $t$. The training is stopped either after a maximum number of iterations $T$ or when a satisfactory tolerance on the gradient norm is reached.

\subsection{Permutation Equivariance}\label{subsec:permutationEquivariance}

GNNs are equivariant to permutations in the support, i.e., the node relabeling. This allows them exploiting the internal symmetries of graph signals during training and being transferable to different graphs during testing \cite{gama2020graphs}. We show next the proposed SR-GNN is permutation equivariant. 
\begin{proposition}\label{thm:permutationEquivariance}
	Consider the SR-GNN $\bbPhi(\bbx, \bbV;\bbS,\ccalH)$ with the graph signal $\bbx$, the underlying graph $\bbS$ and the parameters $\ccalH$. Let $\bbP$ be a permutation matrix, i.e.,
	\begin{align}
		\bbP \in \{0,1\}^{n \times n} : \bbP \bbI = \bbI, \bbP^\top\bbI = \bbI 
	\end{align}
	where $\bbI$ is the identity matrix. Then, it holds that
	\begin{align}
		\bbPhi(\bbP^\top\! \bbx, \bbP^\top\! \bbV; \bbP^\top \bbS \bbP, \ccalH) = \bbP^\top \bbPhi(\bbx, \bbV;\bbS,\ccalH).
	\end{align}
\end{proposition}
\begin{proof}
	We start by considering the graph filter $\bbH(\bbS)$. Denote by $\hat{\bbS} = \bbP^\top \bbS \bbP$ and $\hat{\bbx}=\bbP^\top \bbx$ the concise notations. The output of the graph filter applied on $\hat{\bbx}$ over $\hat{\bbS}$ is given by
	\begin{align}
		\bbH(\hat{\bbS})\hat{\bbx} &= \sum_{k=0}^K h_k \hat{\bbS}^k \hat{\bbx} = \sum_{k=0}^K h_k (\bbP^\top \bbS \bbP)^k \bbP^\top\! \bbx\\
		&= \sum_{k=0}^K h_k \bbP^\top \bbS \bbx = \bbP^\top \sum_{k=0}^K h_k \bbS \bbx= \bbP^\top \bbH(\bbS)\bbx. \nonumber
	\end{align}
	Consider the eigendecomposition $\hat{\bbS} = \hat{\bbV} \hat{\bbLambda} \hat{\bbV}^\top$ with $\hat{\bbV} = \bbP^\top \bbV \bbP$ the eigenvectors and $\hat{\bbLambda} \!=\! \text{diag}(\hat{\lambda}_1,\ldots,\hat{\lambda}_n)\!=\! \bbP^\top \bbLambda \bbP$ the eigenvalues. The filter output applied on $\hat{\bbv}_i$ over $\hat{\bbS}$ is given by
	\begin{align}
		\hat{z}_{i} = \hat{\bbv}_i^\top \bbH(\hat{\bbS})\hat{\bbv}_i &= h(\hat{\lambda}_i),~\forall~i=1,\ldots,n.
	\end{align}
	and we have 
	\begin{align}
		\hat{\bbz} = [\hat{z}_1,\ldots,\hat{z}_n]^\top = \bbP^\top [z_1,\ldots,z_n] = \bbP^\top \bbz.
	\end{align}
	Both are permutations by $\bbP$ of the outputs of the graph filter applied on $\bbx$ and $\{\bbv_i\}_{i=1}^n$. Since the nonlinearity $\sigma(\cdot)$ and the maxout nonlinearity $\max(\cdot)$ are both pointwise, they have no effect on the permutation equivariance. Therefore, if the input to a SR-GNN layer is permuted, its output will be permuted accordingly. As such permutations cascade through multiple layers, the SR-GNN is permutation equivariant, which completes the proof.
\end{proof}
\noindent Theorem \ref{thm:permutationEquivariance} states that the outputs of the SR-GNN are independent of the graph labeling. This indicates that the SR-GNN is capable of harnessing the same information regardless of what permuted versions of the signal and the graph are processed, which yields it robust transference on signals and graphs with similar structures. In the next section, we corroborate the theoretical findings with numerical experiments.

\section{Experiments}\label{sec:experiment}

We evaluate the performance of the SR-GNN with synthetic data from source localization and real data from movie recommendation. The SR-GNN is of $L=2$ layers, each of which comprises $F=32$ graph filters of order $K=5$ and the ReLU nonlinearity. We compare the SR-GNN with the vanilla GNN and an alternative method that regularizes the frequency response without considering the information propagation, referred to as the GNN regularization \cite{gama2020stability1}. We consider the graph perturbations are entirely random during testing and the perturbation information is not available during training. The Adam optimizer is used with a learning rate $\alpha = 10^{-3}$ and decaying factors $\beta_1 = 0.9$, $\beta_2 = 0.999$ \cite{kingma2014adam}. The regularization parameter is $\gamma = 0.1$ and all results are averaged over $10$ random data splits.

\subsection{Source localization} 

The goal of this experiment is to find the source community of a diffused signal. We consider a signal diffusion process over a stochastic block model graph of $n=50$ nodes equally divided into $5$ communities, where the intra- and the inter-community edge probabilities are $0.8$ and $0.2$. The initial signal is a Kronecker delta $\bbdelta_{s} = [\delta_1,\ldots,\delta_n]^\top \in \{0,1\}^n$ with $\delta_s \ne 0$ at the source node $s \in \{s_1,\ldots,s_5\}$. The diffused signal at time instance $t$ is $\bbx^t_{s} = \bbS^t \bbdelta_s + \bbn$, where $\bbS$ is the normalized adjacency matrix and $\bbn$ is the Gaussian noise. The training set consists of $10^4$ samples generated randomly with a source node $s$ and a diffused time $t$, while the validation and testing sets contain $2.5\times 10^3$ samples respectively. We consider the underlying graph as the communication graph, which is perturbed by an error matrix $\bbE$ during testing. The performance is measured with the classification accuracy, which is the ratio of accurately classified signals to total testing signals \cite{gama2020graphs}. 

\smallskip
\noindent \textbf{Convergence.} We first corroborate the convergence of the SR-GNN training procedure. Fig. \ref{fig:convergence} shows that the cost value and the maximal value of the average spectral output $\sum_{\ell=1}^L \max_{i=1,\ldots,n} (z_{\ell,i})/L$ [cf. \eqref{eq:layerOutput}] decrease with the number of iterations, both of which achieve convergent conditions. The former converges close to zero suggesting a satisfactory learning performance. The latter approaches and stabilizes around one, which pursues a trade-off between the perturbation stability and the information propagation -- see Section \ref{subsec:training}.

\smallskip
\noindent \textbf{Performance.} We then evaluate the impact of graph perturbations on the architecture performance. Fig. \ref{fig:source} shows the classification accuracy under different error sizes $\eps \in [0, 0.01]$. For the unperturbed case $\eps = 0$, the GNN and the SR-GNN exhibit comparable performance, among which the former is slightly better. The GNN regularization performs worst because it ignores the signal propagation that results in an information loss. For the perturbed cases, the SR-GNN exhibits the best performance, which demonstrates an improved stability to graph perturbations and corroborates the theoretical analysis in Section \ref{sec:Analysis}. The GNN is more vulnerable under graph perturbations, which is emphasized with the increase of error size $\eps$. The GNN regularization performs worse than the GNN for small error sizes because of the information loss, but outperforms the GNN for large error sizes because of the improved stability. The latter highlights that the SR-GNN achieves a satisfactory balance between these two factors.

\subsection{Movie recommendation} 

The goal of this experiment is to predict the rating a user would give to a movie. We consider a subset of the MovieLens-100k dataset with $943$ users and $400$ movies \cite{Harper2016}. The underlying graph is the movie similarity graph, where the movie similarity is the Pearson correlation and ten edges with highest similarity are kept for each node (movie) \cite{Huang2018}. The graph signal is the ratings of movies given by a user, where the value is zero if that movie is not rated. We similarly consider the underlying graph may be perturbed by an error matrix $\bbE$ during testing. This could be the case where some ratings are prone to adversarial attacks or when the Pearson similarity is measured from a few common ratings (e.g., high sparsity in the user-item matrix). The performance is measured with the root mean square error (RMSE).

\smallskip
\noindent \textbf{Performance.} We measure the architecture performance under different error sizes $\eps \in [0, 0.1]$ in Fig. \ref{fig:movie}. For the unperturbed case $\eps = 0$, the GNN exhibits the best performance. The SR-GNN takes the second place but performs comparably to the GNN. The GNN regularization performs worst because of an information loss caused by the regularization. For the perturbed cases, while all the three methods degrade, the SR-GNN achieves the lowest RMSE and exhibits robustness to graph perturbations. The GNN degrades quickly with the error size $\eps$ and shows vulnerability to graph perturbations.



\section{Conclusions} \label{sec_conclusions}
We developed the self-regularized graph neural network that achieves a trade-off between the learning performance and the stability to topological perturbations. The proposed architecture considers the graph signal and the eigenvectors of the underlying graph as inputs and generates the task-relevant features and the layer frequency responses as outputs. We train the SR-GNN by solving a regularized optimization problem, where the task-relevant cost is regularized by the frequency responses. It decreases the cost to improve performance and incentivizes the maximal frequency response around one to pursue a trade-off between the perturbation stability and the information propagation. We further prove the SR-GNN preserves the permutation equivariance, which guarantees its transference to similar graph structures. Experiment results on source localization and movie recommendation corroborate the effectiveness of the SR-GNN to find a favourable balance between performance and stability.

{\small
	\bibliographystyle{IEEEtran}
	\bibliography{myIEEEabrv,biblioOp}
}

\end{document}